\documentclass{llncs}

\usepackage{amsmath}
\usepackage{amssymb}
\usepackage{color}
\usepackage{epstopdf}
\usepackage{times}

\usepackage{url}
\urldef{\mailsa}\path|{Shenggang.Ying, Yuan.Feng, Mingsheng.Ying}@uts.edu.au|

\usepackage{graphicx}
\usepackage[ruled]{algorithm2e}
\usepackage{dsfont}

\def\bscc{B}

\def\>{\ensuremath{\rangle}}
\def\<{\ensuremath{\langle}}
\def\ra{\ensuremath{\rightarrow}}

\def\e{\ensuremath{\mathcal{E}}}

\def\h{\ensuremath{\mathcal{H}}}
\def\la{\ensuremath{\leftarrow}}

\newcommand{\ket}[1]{|#1\rangle}

\newcommand{\tr}{{\rm tr}}
\newcommand {\spa } {{\rm span}}
\newcommand {\supp } {{\rm supp}}
\newcommand {\I } {{\mathcal I}}

\newcommand{\B}{\mathcal{B}}

\newcommand {\D } {{\mathcal{D}}}
\newcommand {\E } {{\mathcal{E}}}

\newcommand{\hs}{\mathcal{H}}

\newcommand{\RG}{{\mathcal{R}_\mathcal{G}}}
\newcommand{\G}{{\mathcal{G}}}

\begin{document}

\title{Reachability Probabilities of Quantum Markov Chains}

\titlerunning{Quantum Markov Chains}

\author{Shenggang Ying, Yuan Feng, Nengkun Yu \and Mingsheng Ying}
\authorrunning{S. G. Ying, Y. Feng and M. S. Ying}

\institute{Tsinghua University, China\\ University of Technology,
Sydney, Australia\\
\mailsa
}

\toctitle{Quantum Markov Chains}
\tocauthor{Ying, Feng and Ying}
\maketitle

\begin{abstract}

This paper studies three kinds of long-term behaviour, namely reachability, repeated reachability and persistence, of quantum Markov chains (qMCs). As a stepping-stone, we introduce the notion of bottom strongly connected component (BSCC) of a qMC and develop an algorithm for finding BSCC decompositions of the state space of a qMC. As the major contribution, several (classical) algorithms for computing the reachability, repeated reachability and persistence probabilities of a qMC are presented, and their complexities are analysed.
\end{abstract}


\section{Introduction}

Verification problems of quantum systems are emerging from quantum physics, quantum communication and quantum computation. For example, verification has been identified by physicists as one of the major short-term goals of quantum simulation~\cite{CZ12}. Some effective verification techniques for quantum cryptographic protocols have recently been developed \cite{GPN10}, \cite{DAV11}, based on either quantum process algebras~\cite{JL04}, \cite{GN05}, \cite{FDY11}, \cite{FDY12} or quantum model-checking~\cite{GPN08}. Also, several methods for verifying quantum programs \cite{Se04} have been proposed, including quantum weakest preconditions \cite{DP06} and quantum Floyd-Hoare logic \cite{Yin11}.

A quantum Markov chain (qMC) is a quantum generalisation of Markov chain (MC) where, roughly speaking, the state space is a Hilbert space, and the transition probability matrix of a MC is replaced by a super-operator, which is a mathematical formalism of the discrete-time evolution of (open) quantum systems. qMCs have been widely employed as a mathematical model of quantum noise in physics~\cite{GZ04} and as a model of communication channels in quantum information theory~\cite{NC00}. A special class of qMCs, namely quantum walks, has been successfully used in design and analysis of quantum algorithms~\cite{Am03}. Recently, the authors~\cite{YY12} introduced a model of concurrent quantum programs in terms of qMCs as a quantum extension of Hart-Sharir-Pnueli's Markov chain model of probabilistic concurrent programs~\cite{HSP83}. This paper considers the verification problem of qMCs.

Reachability analysis is at the center of verification and model-checking of both classical and probabilistic systems.
Reachability of quantum systems was first studied by physicists~\cite{SSL02} within the theme of quantum control
, but they only considered states reachable in a single step of evolution.
In~\cite{YY12}, reachability of qMCs was considered, and it was used in termination checking of concurrent quantum programs. However, reachability studied in \cite{YY12} can be properly described as \textit{qualitative} reachability because only algorithms for computing reachable subspaces but not reachability probabilities were developed. This paper is a continuation of \cite{YY12} and aims at \textit{quantitative} reachability analysis for qMCs. More precisely, the main purpose  of this paper is to develop (classical) algorithms for computing the reachability, repeated reachability and persistence probabilities of qMCs.

Reachability analysis techniques for classical MCs heavily depends on algorithms for graph-reachability problems, in particular for finding bottom strongly connected components (BSCCs) of the underlying graph of a MC (see \cite[Section 10.1.2]{BK08}). Such algorithms have been intensively studied by the graph algorithms community since early 1970's (see~\cite[Part VI]{Algor}; \cite{Yanna90}), and are ready to be directly adopted in reachability analysis of MCs. However, we don't have the corresponding algorithms for qMCs in hands and have to start from scratch. So, in order to conduct reachability analysis for qMCs we introduce the notion of BSCC and develop an algorithm for finding BSCC decomposition for qMCs in this paper.
Interestingly, there are some essential differences between BSCCs in the classical and quantum cases. For example, BSCC decomposition of a qMC is unnecessary to be unique. Also, classical algorithms for finding BSCCs like depth-first search cannot be directly generalised to qMCs. Instead, it requires very different ideas to develop algorithms for finding BSCCs of qMCs, appealing to matrix operation algorithms~\cite[Chapter 28]{Algor} through matrix representation of super-operators. The major challenge in dealing with quantum BSCCs, which would not arise in classical BSCCs at all, is to maintain the linear algebraic structure underpinning quantum systems. We believe that these results for quantum BSCCs obtained in this paper are also of independent significance.

This paper is organised as follows. The preliminaries are presented in Sec.~\ref{sec:Preliminaries}; in particular we recall the notion of qMC and define the graph structure of a qMC. The notion of BSCC of a qMC is introduced in Sec.~\ref{sec: BSCC}, where a characterisation of quantum BSCC is given in terms of the fixed points of super-operators, and an algorithm for checking whether a subspace of the state Hilbert space of a qMC is a BSCC is given. In Sec.~\ref{sec:de}, we define the notion of transient subspace of a qMC and show that the state space of a qMC can be decomposed into the direct sum of a transient subspace and a family of BSCCs. Furthermore, it is proved that although such a decomposition is not unique, the dimensions of its components are fixed. In particular, an algorithm for constructing BSCC decomposition of qMCs is found. With the preparation in Secs.~\ref{sec: BSCC} and \ref{sec:de}, we examine reachability of a qMC in Sec.~\ref{sec:rb}, where  an algorithm for computing reachability probability is presented. An algorithm for computing repeated reachability and persistence probabilities is finally developed in Sec. \ref{sec:pers}. Sec.~\ref{sec:Conclusions} is a brief conclusion.

\section{Quantum Markov Chains and Their Graph Structures}\label{sec:Preliminaries}

\subsection{Basics of Quantum Theory}

For convenience of the reader, we recall some basic notions from quantum theory; for details we refer to \cite{NC00}. The state space of a quantum system is a Hilbert space. In this paper, we only consider a finite-dimensional Hilbert space $\hs$, which is just a finite-dimensional complex vector space with inner product. The inner product of two vectors $|\phi\>,|\psi\>\in\hs$ is denoted by $\<\phi|\psi\>$. A pure quantum state is a normalised vector $|\phi\rangle$ in $\hs$ with $\<\phi|\phi\> = 1$. We say that two vectors $|\phi\rangle$ and $|\psi\rangle$ are orthogonal, written $|\phi\rangle\bot |\psi\rangle$, if $\langle\phi|\psi\rangle=0$.
A mixed state is represented by a density operator, i.e. a positive operator $\rho$ on $\hs$ with $tr(\rho)=1$, or equivalently a positive semi-definite and trace-one $n\times n$ matrix if $\dim\hs =n$. In particular, for each pure state $|\psi\rangle$, there is a corresponding density operator $\psi=|\psi\rangle\langle\psi|$. For simplicity, we often use pure state $|\psi\rangle$ and density operator $\psi$ interchangeably. A positive operator $\rho$ is called a partial density operator if trace $\tr(\rho)\leq 1$. The set of partial density operators on $\hs$ is denoted by $\mathcal{D}(\hs)$. The support $\supp (\rho)$ of a partial density operator $\rho\in\mathcal{D}(\mathcal{H})$ is defined to be the space spanned by the eigenvectors of $\rho$ with non-zero eigenvalues. The set of all (bounded) operators on $\hs$, i.e. $d\times d$ complex matrices with $d=\dim\hs$, is denoted by $\B(\hs)$.

For any set $V$ of vectors in $\hs$, we write $\spa V$ for the subspace of
$\hs$ spanned by $V$; that is, it consists of all finite linear combinations of vectors in $V$.
Two subspaces $X$ and $Y$ of $\hs$ are said to be orthogonal, written $X\bot Y$, if $|\phi\rangle\bot |\psi\rangle$ for any $|\phi\rangle\in X$ and $|\psi\rangle\in Y$. The ortho-complement $X^{\bot}$ of a subspace $X$ of $\hs$ is the subspace of vectors orthogonal to all vectors in $X$. An operator $P$ is called the projection onto a subspace $X$ if $P\ket{\psi}=\ket{\psi}$ for all $\ket{\psi}\in X$ and $P\ket{\psi}=0$ for all $\ket{\psi}\in X^{\bot}$. We write $P_X$ for the projection onto $X$. According to the theory of quantum measurements, for any density operator $\rho$, trace $tr(P_X\rho)$ is the probability that the mixed state $\rho$ lies in subspace $X$.
Let $\{X_k\}$ be a family of
subspaces of $\hs$. Then the join of $\{X_k\}$ is defined by $$\bigvee_k
X_k=\spa(\bigcup_k X_k).$$ In particular, we write $X\vee Y$ for the join of
two subspaces $X$ and $Y$. It is easy to see that $\bigvee_k X_k$ is the smallest subspace of $\hs$ that contains all $X_k$.

Composed quantum systems are modeled by tensor products. If a quantum system consists of two subsystems with state spaces $\hs_1$ and $\hs_2$, then its state space is $\hs = \hs_1\otimes\hs_2$, which is the Hilbert space spanned by vectors $|\psi_1\rangle|\psi_2\rangle=|\psi_1\rangle\otimes|\psi_2\rangle$ with $|\psi_1\rangle\in\hs_1$ and $|\psi_2\rangle\in\hs_2$. For any operators $A_1$ on $\hs_1$ and $A_2$ on $\hs_2$, their tensor product $A_1\otimes A_2$ is defined by $$(A_1\otimes A_2)(|\psi_1\rangle|\psi_2\rangle)=(A_1|\psi_1\rangle)\otimes(A_2|\psi_2\rangle)$$ for all $|\psi_1\rangle\in\hs_1$ and $|\psi_2\rangle\in\hs_2$ together with linearity.

The evolution of a closed quantum system is described as a unitary operator, i.e. an operator $U$ on $\hs$ with $U^\dag U = UU^\dag = I$, where $I$ is the identity on $\hs$. A pure state $|\phi\>$ becomes $U|\phi\>$ after this unitary evolution $U$, while a mixed state $\rho$ becomes $U\rho U^\dag$. The dynamics of an open quantum system is described by a super-operator, i.e. a linear map
$\E$ from the space of linear operators on $\hs$ into itself, satisfying the following conditions:
\begin{enumerate}
\item $\tr[\E(\rho)]\leq \tr(\rho)$ for all $\rho\in\mathcal{D}(\hs)$, with equality for trace-preserving $\E$;
\item Complete positivity: for any extra Hilbert
space $\mathcal{H}_R$, $(\mathcal{I}_R\otimes \E)(A)$ is positive
provided $A$ is a positive operator on $\mathcal{H}_R\otimes \hs$, where
$\mathcal{I}_R$ is the identity map on the space of linear operators on $\mathcal{\hs}_R$.\end{enumerate}
In this paper, we only consider trace-preserving super-operators. Each super-operator has a Kraus operator-sum representation: $\E=\sum_i E_i\cdot E_i^{\dag}$, or more precisely $$\E(\rho) = \sum E_i \rho E_i^\dag$$ for all $\rho\in\mathcal{D}(\hs)$, where $E_i$ are operators on $\hs$ such that $\sum_i E^\dag E_i=I$.

\subsection{Quantum Markov Chains}\label{Sec:QuantumGraphs}

Now we are ready to introduce the notion of quantum Markov chain. Recall that a Markov chain is a pair $\langle S,P\rangle$, where $S$ is a finite set of states, and $P$ is a matrix of transition probabilities, i.e. a mapping $P: S\times S\rightarrow [0,1]$ such that $$\sum_{t\in S}P(s,t)=1$$ for every $s\in S$, where $P(s,t)$ is the probability of going from $s$ to $t$. A quantum Markov chain is a quantum generalisation of a Markov chain where the state space of a Markov chain is replaced by a Hilbert space and its transition matrix is replaced by a super-operator.

\begin{definition}A quantum Markov chain is a pair $\mathcal{G}=\langle\hs,\E\rangle$, where $\hs$ is a finite-dimensional Hilbert space, and $\E$ is a super-operator on $\hs$.\end{definition}

The behaviour of a quantum Markov chain can be described as follows: if currently the process is in a mixed state $\rho$, then it will be in state $\E(\rho)$ in the next step. Both $\rho$ and $\E(\rho)$ can be written as statistical ensembles:
$$\rho=\sum_i p_i |\phi_i\rangle\langle\phi_i|, \hspace{2em} \E(\rho)=\sum_j q_j |\psi_j\rangle\langle\psi_j|,$$
where $p_i, q_j\geq 0$ for all $i,j$, and $\sum_ip_i=\sum_jq_j=1$. So, super-operator $\E$ can be understood as an operation that transfers statistical ensemble $\{(p_i,|\phi_i\rangle)\}$ to $\{(q_j,|\psi_j\rangle)\}$. In this way, a quantum Markov chain can be seen as a generalisation of a Markov chain.

\subsection{Graphs in Quantum Markov Chains}

There is a natural graph structure underlying a quantum Markov chain. This can be seen clearly by introducing adjacency relation in it. To this end, we first introduce an auxiliary notion. The image of a subspace $X$ of $\hs$ under a super-operator $\E$ is defined to be $$\E(X) = \bigvee_{\ket{\psi}\in X}\supp (\E(\psi)).$$ Intuitively, $\E(X)$ is the subspace of $\hs$ spanned by the images under $\E$ of states in $X$.

\begin{definition} Let $\G = \<\hs, \E\>$ be a quantum Markov chain, and
let $\ket{\varphi}$ and $\ket{\psi}$ be pure states and $\rho$ and $\sigma$ mixed states in $\hs$. Then
\begin{enumerate}
\item $\ket{\varphi}$ is adjacent to $\ket{\psi}$ in $\mathcal{G}$, written $\ket{\psi}\rightarrow\ket{\varphi}$, if $\ket{\varphi}\in\E(X_\psi)$, where $X_\psi=\spa\{|\psi\rangle\}$. \item $\ket{\varphi}$ is adjacent to  $\rho$, written $\rho\rightarrow\ket{\varphi}$, if $\ket{\varphi}\in\E(\supp (\rho))$. \item $\sigma$ is adjacent to $\rho$, written $\rho\rightarrow\sigma$, if $\supp(\sigma)\subseteq\E(\supp (\rho))$.\end{enumerate}\end{definition}

\begin{definition}\begin{enumerate}\item
    A sequence $\pi=\rho_0\rightarrow\rho_1\rightarrow\cdot\cdot\cdot\rightarrow\rho_n$ of adjacent density operators in a quantum Markov chain $\G$ is called a path from $\rho_0$ to $\rho_n$ in $\G$, and its length is $|\pi|=n$.\item For any density operators $\rho$ and $\sigma$, if there is a path from $\rho$ to $\sigma$ then we say that $\sigma$ is reachable from $\rho$ in $\mathcal{G}$.
\end{enumerate}\end{definition}

\begin{definition} Let $\G = \<\hs, \E\>$ be a quantum Markov chain. For any $\rho\in \D(\hs)$, its reachable space in $\mathcal{G}$ is $$\mathcal{R}_{\mathcal{G}}(\rho)=\spa\{|\psi\rangle\in\mathcal{H}:|\psi\rangle\ {\rm is\ reachable\ from}\ \rho\ {\rm in}\ \mathcal{G}\}.$$
\end{definition}

The following lemma is very useful for our later discussion.
\begin{lemma}\label{lem:TransitiveRelation}\begin{enumerate} \item (Transitivity of reachability)
For any $\rho,\sigma\in \D(\hs)$, if $\supp(\rho)\subseteq \mathcal{R}_\mathcal{G}(\sigma)$, then $\mathcal{R}_\mathcal{G}(\rho) \subseteq \mathcal{R}_\mathcal{G}(\sigma)$.
\item \cite[Theorem 1]{YY12}\label{thm:reachableSpace}
If $d=\dim\mathcal{H}$, then for any $\rho\in \D(\hs)$, we have
\begin{equation}
    \mathcal{R}_{\mathcal{G}}(\rho)=\bigvee_{i=0}^{d-1} \supp(\E^i(\rho)).
\end{equation}\end{enumerate}
\end{lemma}

\section{Bottom Strongly Connected Components}\label{sec: BSCC}

\subsection{Basic Definitions}

The notion of bottom strongly connected component plays an important role in model checking Markov chains. In this section, we extend this notion to the quantum case. We first introduce an auxiliary notation. Let $X$ be a subspace of a Hilbert space, and let $\E$ be a super-operator on $\hs$. Then the restriction of $\E$ on $X$ is defined to be super-operator $\E|_X$ with $$\E|_X(\rho)=P_X\E(\rho)P_X$$ for all $\rho\in\D(X)$, where $P_X$ is the projection onto $X$.

\begin{definition} Let $\G=\langle\hs,\E\rangle$ be a quantum Markov chain. A subspace $X$ of $\hs$ is called strongly connected in $\G$ if for any $\ket{\varphi},\ket{\psi}\in X$, we have $\ket{\varphi} \in\mathcal{R}_{\mathcal{G}_X}(\psi)$ and $\ket{\psi}\in\mathcal{R}_{\mathcal{G}_X}(\varphi)$, where quantum Markov chain $\G_X=\langle X,\E_X\rangle$ is the restriction of $\G$ on $X$.\end{definition}

We write $SC(\mathcal{G})$ for the set of strongly connected subspaces of $\hs$ in $\mathcal{G}.$ It is easy to see that $(SC(\mathcal{G}),\subseteq)$ is an inductive set; that is, for any subset $\{X_i\}$ of $SC(\mathcal{G})$ that is linearly ordered by $\subseteq$, we have $\bigcup_iX_i\in SC(\mathcal{G})$. Thus, by Zorn lemma we assert that there exists a maximal element in $SC(\mathcal{G})$.

\begin{definition}A maximal element of $(SC(\mathcal{G}),\subseteq)$ is called a strongly connected component (SCC) of $\mathcal{G}$.\end{definition}

To define bottom strongly connected component, we need an auxiliary notion of invariant subspace.

\begin{definition}Let $\G=\langle\hs,\E\rangle$ be a quantum Markov chain. Then a subspace $X$ of $\hs$ is said to be invariant in $\G$ if $\E(X)\subseteq X$.\end{definition}

It is easy to see that if super-operator $\E$ has the Kraus representation $\E=\sum_i E_i\cdot E_i^\dag$, then $X$ is invariant if and only if $E_i X \subseteq X$ for all $i$. Recall that in a classical Markov chain, the probability of staying in an invariant subset is non-decreasing. A quantum generalisation of this fact is presented in the following:

\begin{theorem}\label{thm:SPnondec}
    For any invariant subspace $X$ of $\hs$ in a quantum Markov chain $\G=\langle\hs,\E\rangle$, we have $$tr(P_X\E(\rho)) \geq tr(P_X\rho)$$ for all $\rho\in\D(\hs)$, where $P_X$ is the projection onto $X$.
\end{theorem}

Now we are ready to introduce the key notion of this section.

\begin{definition}
Let $\G=\langle\hs,\E\rangle$ be a quantum Markov chain. Then a subspace $X$ of $\hs$ is called a bottom strongly connected component (BSCC) of $\G$ if it is a SCC of $\G$ and invariant in $\G$.
\end{definition}

\begin{example}\label{exam:E5} Consider quantum Markov chain $\G=\langle\hs,\E\rangle$ with state space
     $\hs=\spa\{|0\>,$ $\cdots,|4\>\}$ and super-operator $$\E=\sum_{i=1}^5 E_i\cdot E_i^\dag,$$ where the operators $E_i$ (i=1,...,5) are given as follows:
    \begin{eqnarray*}
    E_1 &=& \frac{1}{\sqrt{2}}(|1\>\<0 +1| + |3\>\<2+3|),\ \ \   E_2 =\frac{1}{\sqrt{2}}(|1\>\<0-1| + |3\>\<2-3|),\\
    E_3 &=&\frac{1}{\sqrt{2}}(|0\>\<0+1| + |2\>\<2+3|),\ \ \   E_4 =\frac{1}{\sqrt{2}}(|0\>\<0-1| + |2\>\<2-3|),\\
    E_5 &=&\frac{1}{10}(|0\>\<4| + |1\>\<4| + |2\>\<4| + 4 |3\>\<4| + 9 |4\>\<4|),
    \end{eqnarray*}
and the states used above are defined by $$|0\pm 1\>= (|0\>\pm |1\>)/\sqrt{2}\ {\rm and}\ |2\pm 3\>= (|2\>\pm |3\>)/\sqrt{2}.$$ It is easy to see that $\bscc=\spa\{|0\>,|1\>\}$ is a BSCC of quantum Markov chain $\G$, as for any $|\psi\>=\alpha|0\>+\beta|1\>\in \bscc$, we have $\E(\psi) = (|0\>\<0|+|1\>\<1|)/2$.
\end{example}

The following lemma clarifies the relationship between different BSCCs.
\begin{lemma}\label{lem:BSCCRG}\begin{enumerate}
\item For any two different BSCCs $X$ and $Y$ of quantum Markov chain $\G$, we have $X \cap Y = \{0\}$ ($0$-dimensional Hilbert space).
\item If $X$ and $Y$ are two BSCCs of $\G$ with $\dim X\neq \dim Y$, then $X\bot Y$.
\end{enumerate}
\end{lemma}

\subsection{Characterisations of BSCCs}

This subsection purports to give two characterisations of BSCCs. The first is presented in terms of reachable spaces.

\begin{lemma}\label{lem:BSCCRG1} A subspace $X$ is a BSCC of quantum Markov chain $\G$ if and only if $\RG(\phi)  = X$ for any non-zero $|\phi\>\in X$.
\end{lemma}

To present the second characterisation, we need the notion of fixed point of super-operator.

\begin{definition}\begin{enumerate}\item A nonzero partial density operator $\rho\in \mathcal{D}(\h)$ is called a fixed point state of super-operator $\E$ if $\E(\rho) = \rho$.
 \item A fixed point state $\rho$ of super-operator $\E$ is called minimal if for any fixed point state $\sigma$ of $\E$, it holds that $supp(\sigma) \subseteq supp(\rho)$ implies $\sigma=\rho$.\end{enumerate}
\end{definition}

The second characterisation of BSCCs establishes a connection between BSCCs and minimal fixed point states.

\begin{theorem}\label{thm:BSCC}
    A subspace $X$ is a BSCC of quantum Markov chain $\G=\langle\hs,\E\rangle$ if and only if there exists a minimal fixed point state $\rho$ of $\E$ such that $supp(\rho) = X$. Furthermore, $\rho$ is actually the \emph{unique} fixed point state, up to normalisation, with the support included in $X$.
\end{theorem}

\subsection{Checking BSCCs}

We now present an algorithm that decides whether or not a given subspace is a BSCC of a quantum Markov chain (see Algorithm \ref{alg:checkBSCC}). The correctness and complexity of this algorithm are given in the following theorem.

\begin{algorithm}
    \SetKwData{Left}{left}\SetKwData{This}{this}\SetKwData{Up}{up}
    \SetKwFunction{Union}{Union}\SetKwFunction{FindCompress}{FindCompress}
    \SetKwInOut{Input}{input}\SetKwInOut{Output}{output}

    \Input{A quantum Markov chain $\G=\<\hs, \E\>$ and a subspace $X\subseteq \hs$}
    \Output{True or False indicating whether $X$ is a BSCC of $\G$}

    \Begin{
        \If{$\E(X) \not\subseteq X$}{
            \Return{False}\;
        }
        $\E' \la P_X\circ\E$\;
        $\B \la \text{a density operator basis of the set}~ \{A\in \B(\hs) : \E'(A) = A\}$;\hfill (*)\\
        \eIf{$|\B| > 1$}{
            \Return{False}\;
        }{
            $\rho \la \text{the unique element in} ~\B$\;
            \eIf{$X = \supp(\rho)$}{
                \Return{True}\;
            }{
                \Return{False}\;
            }
        }
    }
    \caption{CheckBSCC($X$)}\label{alg:checkBSCC}
\end{algorithm}

\begin{theorem}\label{thm:algcheckBSCC}
Given a quantum Markov chain $\<\hs, \E\>$ and a subspace  $X\subseteq \hs$, Algorithm~\ref{alg:checkBSCC} decides whether or not $X$ is a BSCC of $\G$ in time $O(n^6)$, where $n=\dim(\hs)$.\end{theorem}

\section{Decompositions of the State Space}\label{sec:de}

A state in a classical Markov chain is transient if there is a non-zero probability that the process will never return to it, and a state is recurrent if from it the returning probability is 1.
It is well-known that a state is recurrent if and only if it belongs to some BSCC in a finite-state Markov chain,
and thus the state space of a classical Markov chain can be decomposed into the union of some BSCCs and a transient subspace~\cite{BK08}, \cite{MU05}.
The aim of this section is to prove a quantum generalisation of this result.

\begin{definition}
A subspace $X\subseteq \hs$ is transient in a quantum Markov chain $\G=\langle\hs,\E\rangle$ if
    \begin{equation*}
        \lim_{k\rightarrow \infty}\tr(P_X\E^k(\rho)) = 0
    \end{equation*}
for any $\rho \in \D(\hs)$, where $P_X$ is the projection onto $X$.
\end{definition}

The above definition is stated in a \textquotedblleft double negation" way. Intuitively, it means that the probability in a transient subspace will be eventually zero. To understand this definition better, let us recall that in a classical Markov chain, a state $s$ is said to be transient if the system starting from $s$ will eventually return to $s$ with probability less than 1. It is well-known that in a finite-state Markov chain, this is equivalent to that the probability at this state will eventually become 0. In the quantum case, the property ``eventually return'' can be hardly described without measurements, and measurements will disturb the behaviour of the systems. So, we choose to adopt the above definition.

To give a characterisation of transient subspaces, we need the notion of the asymptotic average of a super-operator $\E$, which is defined to be \begin{equation}\label{eq:Einfty}
        \E_\infty = \lim_{N\rightarrow \infty} \frac{1}{N} \sum_{n=1}^N \E^n.
    \end{equation}
It is easy to see from \cite[Proposition 6.3, Proposition 6.9]{Wolf12} that $ \E_\infty$ is a super-operator as well.

\begin{theorem}\label{thm:Flimit}
The ortho-complement of the image of the state space $\hs$ of a quantum Markov chain $\G=\langle\hs,\E\rangle$ under the asymptotic average of super-operator $\E$: $$T_\E:=\E_\infty(\hs)^\perp$$ is the largest transient subspace in $\G$; that is, any transient subspace of $\G$ is a subspace of $T_\E$.
\end{theorem}

We now turn to examine the structure of the image of the state space $\hs$ under super-operator $\E$.
\begin{theorem}\label{thm:Hdec}
   Let $\mathcal{G} = \<\hs,\E\>$ be a quantum Markov chain. Then $\E_\infty(\hs)$ can be decomposed into the direct sum of some orthogonal BSCCs of $\G$.
\end{theorem}

Combining Theorems~\ref{thm:Flimit} and~\ref{thm:Hdec}, we see that the state space of a quantum Markov chain $\G=\langle\hs,\E\rangle$ can be decomposed into the direct sum of a transient subspace of a family of BSCCs:
    \begin{equation}\label{eq:Hdec}
        \hs = \bscc_1\oplus \cdots \oplus \bscc_u \oplus T_\E\end{equation}
where $\bscc_i$'s are orthogonal BSCCs of $\G$.
A similar decomposition was recently obtained in \cite{Rosmanis12} for a special case of $\E^2=\E$. The above decomposition holds for any super-operator $\E$ and thus considerably generalises the corresponding result in \cite{Rosmanis12}.

The BSCC and transient subspace decomposition of a classical Markov chain is unique. However, it is not the case for quantum Markov chains; a trivial example is that $\E$ is the identity operator, for which any 1-dimensional subspace of $\hs$ is a BSCC, and thus for each orthonormal basis $\{|i\>\}$ of $\hs$, $\bigoplus_{i} \spa\{|i\>\}$ is an orthogonal decomposition of $\hs$. The following is a more interesting example.

\begin{example}\label{exam:E5dec} Let quantum Markov chain $\G=\langle\e,\hs\rangle$ be given as in Example \ref{exam:E5}. Then
$\bscc_1 = \spa\{|0\>,|1\>\}$, $\bscc_2 = \spa\{|2\>,|3\>\}$, $D_1 = \spa\{|0+2\>, |1 +3\>\}$, and $D_2 = \spa\{|0-2\>,|1-3\>\}$ are BSCCs, and $T_\E= \spa\{|4\>\}$ is a transient subspace. Furthermore, we have
$$\hs = \bscc_1\oplus\bscc_2 \oplus T_\E= D_1\oplus D_2 \oplus  T_\E.$$
\end{example}

The relation between different decompositions of a quantum Markov chain is clarified by the following theorem.

\begin{theorem}\label{thm:BSCCUnitary}
    Let $\mathcal{G} = \<\hs,\E\>$ be a quantum Markov chain, and let
    \begin{equation*}
        \hs = \bscc_1\oplus \cdots \oplus \bscc_u \oplus T_\E= D_1\oplus \cdots \oplus D_v \oplus T_\E
    \end{equation*}
    be two decompositions in the form of Eq.~(\ref{eq:Hdec}), and $\bscc_i$'s and $D_i$'s are arranged, respectively, according to the increasing order of the dimensions. Then $u=v$, and $\dim(\bscc_i) = \dim(D_i)$ for each $1\leq i\leq u$.
\end{theorem}

To conclude this section, we present an algorithm for finding a BSCC and transient subspace decomposition of a quantum Markov chain (see Algorithm~\ref{al2}).

\begin{algorithm}
    \SetKwData{Left}{left}\SetKwData{This}{this}\SetKwData{Up}{up}
    \SetKwFunction{Union}{Union}\SetKwFunction{FindCompress}{FindCompress}
    \SetKwInOut{Input}{input}\SetKwInOut{Output}{output}
    \Input{A quantum Markov chain $\G=\<\hs, \E\>$}
    \Output{A set of orthogonal BSCCs $\{\bscc_i\}$ and a transient subspace $T_\E$ such that $\hs = \bigoplus_i \bscc_i \oplus T_\E$}

    \Begin{
        $\mathcal{B} \la$ Decompose($\E_\infty(\h)$)\;
        \Return{$\mathcal{B}$, $\E_\infty(\h)^\perp$}\;
    }
    \caption{DecomposeH($\G$)}\label{al2}
\end{algorithm}

\begin{theorem}\label{thm:algdecomposition}
Given a quantum Markov chain $\<\hs, \E\>$, Algorithm~\ref{al2} decomposes the Hilbert space $\h$ into the direct sum of a family of orthogonal BSCCs and a transient subspace of $\G$ in time $O(n^8)$, where $n=\dim(\hs)$.
\end{theorem}

\section{Reachability Probabilities}\label{sec:rb}
The traditional way to define reachability probabilities in classical Markov chains is first introducing a probability measure based on cylinder sets of finite paths of states. The probability of reaching a set $T$ is then the probability measure of the set of paths which include a state from $T$. Typically, reachability probabilities can be obtained by solving a system of linear equations, which is easy and numerically efficient. In quantum Markov chains, however,
it is even not clear how to define such a probability measure. Thus it seems hopeless to extend reachability analysis to the quantum case in this way.

Fortunately, there is another way to compute the reachability probability in a classical Markov chain $\<S, P\>$. Given a set of states $T\subseteq S$, we first change the original Markov chain into a new one $\<S, P'\>$ by making states in $T$ absorbing. Then
the reachability probability of $T$ is simply the limit of the probability accumulated in $T$, when the time goes to infinity. It turns out that this equivalent definition can be extended into the quantum case as follows.
\begin{definition}
Let $ \<\hs,\E\>$ be a quantum Markov chain, $\rho\in \D(\hs)$ an initial state, and $G\subseteq \hs$ a subspace. Then the probability of reaching $G$, starting from $\rho$, can be defined as
            \begin{equation*}
                 \Pr(\rho\vDash \Diamond G)  = \lim_{i\rightarrow \infty} \tr(P_G \widetilde{\E}^i(\rho))
            \end{equation*}
where $\widetilde{\E} = P_G + \E\circ (I-P_G)$ is the super-operator which first performs the projective measurement $\{P_G, I-P_G\}$ and then applies the identity operator $\I$ or $\E$ depending on the measurement outcome.
\end{definition}
 Obviously the limit in the above definition exists, as the probabilities $\tr(P_G \widetilde{\E}^i(\rho))$ are nondecreasing in $i$.

\begin{procedure}
    \SetKwData{Left}{left}\SetKwData{This}{this}\SetKwData{Up}{up}
    \SetKwFunction{Union}{Union}\SetKwFunction{FindCompress}{FindCompress}
    \SetKwInOut{Input}{input}\SetKwInOut{Output}{output}

    \Input{A subspace $X$ which is the support of a fixed point state of $\E$}
    \Output{A set of orthogonal BSCCs $\{\bscc_i\}$ such that $X = \oplus \bscc_i$}
    \caption{Decompose($X$)}\label{alg:decomposeRho}
    \Begin{
        $\E' \la P_X\circ \E$\;
       $\B \la \text{a density operator basis of the set}~ \{A\in \B(\hs) : \E'(A) = A\}$\;
        \eIf{$|\B| = 1$}{
            $\rho\la \text{the unique element of } \B$\;
            \Return{$\{\supp(\rho)\}$}\;
        }{
            $\rho_1, \rho_2\la \text{two arbitrary elements of } \B$\;
            $\rho \la \text{positive part of} ~\rho_1-\rho_2$\;
            $Y\la \supp(\rho)^\perp$;\hfill (* the ortho-complement of $\supp(\rho)$ in $X$*)\\
            \Return{$\mathrm{Decompose}(\supp(\rho))\cup \mathrm{Decompose}(Y)$}\;
        }
    }
\end{procedure}
To compute the reachability probability, we first note the subspace $G$ is invariant under $\widetilde{\E}$. Thus $ \<G,\widetilde{\E}\>$ is again a quantum Markov chain. Since $\widetilde{\E}(I_G)=I_G$ and $\widetilde{\E}_\infty(G) = G$, we can decompose $G$ into a set of orthogonal BSCCs according to $\widetilde{\E}$ by Theorem~\ref{thm:Hdec}. The following lemma shows a connection between the limit probability of hitting a BSCC and the probability that the asymptotic average of the initial state lies in the same BSCC.

\begin{lemma}\label{lem:BSCClimit}
    Let $\{\bscc_i\}$ be a BSCC decomposition of $\E_\infty(\hs)$, and $P_{\bscc_i}$ the projection onto $\bscc_i$. Then for each $i$, we have
    \begin{equation}\label{eq:BSCClimit}
        \lim_{k\rightarrow \infty} \tr( P_{\bscc_i}\E^k(\rho) ) = \tr(P_{\bscc_i}\E_\infty (\rho))
    \end{equation} for all $\rho\in \D(\hs)$.
\end{lemma}

Lemma~\ref{lem:BSCClimit} and Theorem~\ref{thm:Hdec} together give us an efficient way to compute the reachability probability from a quantum state to a subspace.

\begin{theorem}\label{thm:reach}
Let $ \<\hs,\E\>$ be a quantum Markov chain, $\rho\in \D(\hs)$, and $G\subseteq \hs$ a subspace. Then
            \begin{equation*}
                 \Pr(\rho\vDash \Diamond G)  =  \tr(P_{G}\widetilde{\E}_\infty (\rho)),
            \end{equation*}
and this probability can be computed in time $O(n^8)$ where $n=\dim(\h)$.
\end{theorem}

Our next results assert that if a quantum Markov chain starts from a pure state in a BSCC then its evolution sequence $\psi,\E(\psi),\E^2(\psi),\cdots$ will hit a subspace with non-zero probability infinitely often provided $X$ is not orthogonal to that BSCC. They establishes indeed a certain fairness and thus can be seen as quantum generalisations of Theorems 10.25 and 10.27 in~\cite{BK08}. It is well-known that in the quantum world a measurement will change the state of the measured system. Consequently, fairness naturally splits into two different versions in quantum Markov chains.
\begin{lemma}\label{lem:weakFairness} (Measure-once fairness)
 Let $\bscc$ be a BSCC of quantum Markov chain $\G=\langle\hs,\E\rangle$, and $X$ a subspace which is not orthogonal to $\bscc$. Then for any  $|\psi\>\in\bscc$, it holds that
$\tr(P_X\E^i(\psi)) > 0$ for infinitely many $i$.
\end{lemma}

\begin{lemma}\label{lem:strongFairness} (Measure-many fairness)
 Let $\bscc$ be a BSCC of a quantum Markov chain $\G=\langle\hs,\E\rangle$,
 and $X\subseteq B$ a subspace of $B$.
 Then for any $|\psi\>\in\bscc$, we have $$\lim_{i\rightarrow \infty} \tr(\widetilde{\E}^i(\psi)) = 0,$$
    where $\widetilde{\E} =P_{X^\perp}\circ\E$, and  $X^\perp$ is the ortho-complement of $X$ in $\h$.
\end{lemma}

Lemma~\ref{lem:strongFairness} is stated also in a  \textquotedblleft double negation" way. To best understand it, let us assume that at each step after $\E$ is applied, we perform a projective measurement $\{P_X, P_{X^\bot}\}$. If the outcome corresponding to $P_X$ is observed, the process terminates immediately; otherwise, it continues with another round of applying $\E$. Lemma~\ref{lem:strongFairness} asserts that the probability of nontermination is asymptotically 0; in other words, if we set $X$ as an absorbing boundary, which is included in BSCC $B$, the reachability probability will be absorbed eventually. This lemma is indeed a strong version of fairness. Furthermore, we have:

\begin{theorem}\label{thm:Fairness}
Let $\G=\langle\hs,\E\rangle$ be a quantum Markov chain, and let $X$ be a subspace of $\hs$, and $\widetilde{\E} =P_{X^\perp}\circ\E$. Then the following two statements are equivalent: \begin{enumerate}\item The subspace $X^\bot$ contains no BSCC; \item For any $\rho\in\D(\h)$, we have $$\lim_{i\rightarrow \infty} \tr(\widetilde{\E}^i(\rho)) = 0.$$\end{enumerate}\end{theorem}

It is worth noting that in Theorem \ref{thm:Fairness}, $X$ is not required to be a subspace of a BSCC $\bscc$. The following two examples give some simple applications of Theorem \ref{thm:Fairness}.

\begin{example}
    Consider a quantum walk on an $n$-size cycle~\cite{Am03}. The state space of the whole system is $\hs = \hs_p\otimes\hs_c$, where $\hs_p = \spa\{|0\>,\cdots,|n-1\>\}$ is the position space, and $\hs_c = \spa \{|0\>,|1\>\}$ is the coin space. The evolution of the systems is described by a unitary transformation $U = S(I\otimes H)$, where the coin operator $H$ is the Hadamard operator, and the shift operator
$$S = \sum_{i=0}^{n-1} \left(|i+1\>\<i|\otimes |0\>\<0|+|i-1\>\<i|\otimes |1\>\<1|\right)$$
where the arithmetic operations over the index set are understood as modulo $n$.
If we set absorbing boundaries at position $0$, then from any initial state $|\psi\>$, we know from Theorem~\ref{thm:Fairness} that the probability of nontermination is asymptotically $0$ because there is no BSCC which is orthogonal to the absorbing boundaries.
\end{example}

\begin{example}
    Consider the quantum Markov chain in Example \ref{exam:E5}. Let $\rho_0$ be the initial state, and assume that projective measurement $\{P_0 = |0\>\<0|, P_1 = I-P_0\}$ will be performed at the end of each step and $P_0$ is set as the absorbing boundary. We write $\tilde{\rho}_k = \tilde{\E}^k(\rho_0)$ for the partial density operator after $k$ steps, where $\tilde{\E}=P_1\circ\E.$
    \begin{enumerate}
        \item If $\rho_0 = |1\>\<1|$, then $\lim_{k\rightarrow \infty}\tilde{\rho}_k = 0$. This means the probability will be eventually absorbed.
        \item If $\rho_0 = |2\>\<2|$, then $\lim_{k\rightarrow \infty}\tilde{\rho}_k = (|2\>\<2|+|3\>\<3|)/2$. No probability is absorbed. Let $D_1$ and $D_2$ be as in Example~\ref{exam:E5dec}. Then the probabilities in $D_1$ and $D_2$ are both 0.5. This means that if $\supp(P_0)$ is not totally in a BSCC $D$, then the probability in $D$ may not be absorbed.
    \end{enumerate}
\end{example}

\section{Repeated Reachability and Persistence Probabilities}\label{sec:pers}

In this section, we consider how to compute two kinds of reachability probabilities, namely \textquotedblleft repeated reachability" and \textquotedblleft persistence property", in a quantum Markov chain.
Note that $\E_\infty(\hs)^\perp$ is a transient subspace. We can focus our attention on $\E_\infty(\hs)$.

\begin{definition} Let $\G=\langle\hs,\E\rangle$ be a quantum Markov chain and $G$ a subspace of $\E_\infty(\hs)$.
\begin{enumerate}
\item The set of states in $\E_\infty(\hs)$ satisfying the repeated reachability ``\emph{infinitely often} reaching $G$" is
\begin{equation*}
     \mathcal{X}(G) =  \{|\psi\>\in \E_\infty(\hs): \lim_{k\ra \infty}\tr((P_{G^\perp}\circ\E)^k(\psi))=0 \}.
\end{equation*}
\item The set of states in $\E_\infty(\hs)$ satisfying the persistence property ``\emph{eventually always} in $X$" is
\begin{equation*}
     \mathcal{Y}(G) = \{ |\psi\>\in \E_\infty(\hs) : (\exists N\geq 0)(\forall k\geq N)\ \supp(\E^k(\psi)) \subseteq G\}.
\end{equation*}\end{enumerate}
\end{definition}

The set $\mathcal{X}(G)$ is defined again in a \textquotedblleft double negation" way. Its intuitive meaning can be understood as follows: if the process starts in a state in $\mathcal{X}(G)$ and we make $G$ absorbing, then the probability will be eventually absorbed by $G$.

The following theorem gives a characterisation of $\mathcal{X}(G)$ and $\mathcal{Y}(G)$ and also clarifies the relationship between them.
\begin{theorem}\label{thm:xGyG} For any subspace $G$ of $\E_\infty(\hs)$, both $\mathcal{X}(G)$ and $\mathcal{Y}(G)$ are subspaces of $\hs$. Furthermore, we have
    \begin{equation*}
       \mathcal{X}(G) =  \E_\infty(G), \hspace{2em} \mathcal{Y}(G) =\bigvee_{\bscc\subseteq G} \bscc =  \mathcal{X}(G^\perp)^\perp,
    \end{equation*}
where $\bscc$ ranges over all BSCCs, and the orthogonal complements are taken in $\E_\infty(\hs)$. Moreover, both $\mathcal{X}(G)$ and $\mathcal{Y}(G)$ are invariant.
\end{theorem}

\begin{example}\label{exam:E5xGyGcal}
Let us revisit Example \ref{exam:E5} where $\E_\infty(\hs) = \spa\{|0\>,|1\>,|2\>,|3\>\}$.    \begin{enumerate}
    \item If $G=\spa\{|0\>,|1\>,|2\>\}$, then $\E_\infty(G^\perp)=\supp(\E_\infty(|3\>\<3|)) = \supp((|2\>\<2|+|3\>\<3|)/2)$ and $\E_\infty(G) =\E_\infty(\hs)$. Thus $\mathcal{Y}(G)=\bscc_1$ and $\mathcal{X}(G)=\E_\infty(\hs)$.
    \item If $G = \spa\{|3\>\}$, then $\E_\infty(G^\perp) = \bscc_1\oplus\bscc_2$ and $\E_\infty(G)=\bscc_2$. Thus $\mathcal{Y}(G) = \{0\}$ and $\mathcal{X}(G) = \bscc_2$.
     \end{enumerate}
\end{example}

Now we can define probabilistic persistence and probabilistic repeated reachability.
\begin{definition}
    \begin{enumerate}
     \item The probability that a state $\rho$ satisfies the repeated reachability $\text{rep}(G)$ is the eventual probability in $\mathcal{X}(G)$, starting from $\rho$:
            \begin{equation*}
                \Pr(\rho\vDash \text{rep}(G)) = \lim_{k\rightarrow \infty}\tr(P_{\mathcal{X}(G)} \E^k(\rho)).
            \end{equation*}
        \item The probability that a state $\rho$ satisfies the persistence property $\text{pers}(G)$ is the eventual probability in $\mathcal{Y}(G)$, starting from $\rho$:
            \begin{equation*}
                \Pr(\rho\vDash \text{pers}(G))  = \lim_{k\rightarrow \infty}\tr(P_{\mathcal{Y}(G)} \E^k(\rho)).
            \end{equation*}
           \end{enumerate}
\end{definition}

The well-definedness of the above definition comes from the fact that $\mathcal{X}(G)$ and $\mathcal{Y}(G)$ are invariant. By Theorem~\ref{thm:SPnondec} we know that the two sequences $\{\tr(P_{\mathcal{X}(G)} \E^k(\rho))\}$ and $\{\tr(P_{\mathcal{Y}(G)} \E^k(\rho))\}$ are non-decreasing, and thus their limits exist.
Combining Theorems \ref{thm:Flimit} and \ref{thm:xGyG}, we have:
\begin{theorem}\label{thm:probPR}\begin{enumerate}
\item The repeated reachability probability is
    \begin{equation*}
        \Pr(\rho\vDash \text{rep}(G)) = 1-\tr(P_{\mathcal{X}(G)^\perp}\E_\infty(\rho)) = 1-\Pr(\rho\vDash \text{pers}(G^\perp)).
\end{equation*}
\item
    The persistence probability is
    \begin{equation*}
        \Pr(\rho\vDash \text{pers}(G)) = \tr(P_{\mathcal{Y}(G)}\E_\infty(\rho)).
        \end{equation*}\end{enumerate}
\end{theorem}

Finally, we are able to give an algorithm for computing reachability and persistence probabilities (see Algorithm~\ref{alg:pers}).
\begin{algorithm}
    \SetKwData{Left}{left}\SetKwData{This}{this}\SetKwData{Up}{up}
    \SetKwFunction{Union}{Union}\SetKwFunction{FindCompress}{FindCompress}
    \SetKwInOut{Input}{input}\SetKwInOut{Output}{output}

    \Input{A quantum Markov chain $\<\hs, \E\>$, a subspace $G\subseteq \hs$, and an initial state $\rho\in \D(\hs)$}
    \Output{The probability $\Pr(\rho\vDash \text{pers}(G))$}

    \Begin{
        $\rho_\infty \la \E_\infty(\rho)$\;
        $Y \la \E_\infty(G^\perp)$\;
        $P\la \text{ the projection onto } Y^\perp$; \hfill (* $Y^\perp$ is the ortho-complement of $Y$ in $\E_\infty(\h)$ *)\\
        \Return{$\tr(P\rho_\infty)$}\;
    }
    \caption{Persistence($G$, $\rho$)}\label{alg:pers}
\end{algorithm}

\begin{theorem}\label{thm:algpersistence}
Given a quantum Markov chain $\<\hs, \E\>$, an initial state $\rho\in \D(\h)$, and a subspace $G\subseteq \hs$, Algorithm~\ref{alg:pers} computes persistence probability $\Pr(\rho\vDash \text{pers}(G))$ in time $O(n^8)$, where $n=\dim(\hs)$.
\end{theorem}

With Theorem~\ref{thm:probPR}, Algorithm \ref{alg:pers} can also be used to compute repeated reachability probability $\Pr(\rho\vDash \text{rep}(G))$.

\section{Conclusions}\label{sec:Conclusions}

We introduced the notion of bottom strongly connected component (BSCC) of a quantum Markov chain (qMC) and studied the BSCC decompositions of qMCs. This enables us to develop an efficient algorithm for computing repeated reachability and persistence probabilities of qMCs. Such an algorithm may be used to verify safety and liveness properties of physical systems produced in quantum engineering and quantum programs for future quantum computers.

\newpage

\appendix
\section{Proofs of Lemmas and Theorems}
 We first collect some simple properties of the supports of super-operators for our latter use.

 \begin{proposition}\label{prop:sup} \begin{enumerate}\item If $A=\sum_k\lambda_k |\phi_k\rangle\langle\phi_k|$ where all $\lambda_k>0$ (but $|\phi_k\rangle$'s are not required to be pairwise orthogonal), then $\supp (A)=\spa\{|\phi_k\rangle\};$
\item $\supp(\mathcal{E}(\rho+\sigma))=\supp(\mathcal{E}(\rho))\vee \supp(\mathcal{E}(\sigma));$
\item If $\E=\sum_{i\in I} E_i\cdot E_i^\dag$, then $\E(X)=\spa \{E_i|\psi\rangle: i\in I, |\psi\rangle\in X\};$
\item  $\E(X_1 \vee X_2)= \E(X_1) \vee\E(X_2)$. Thus, $X\subseteq Y\Rightarrow \E(X)\subseteq \E(Y).$
\end{enumerate}
\end{proposition}

Let $\E=\sum_i E_i\cdot E_i^\dag$ be a super-operator on an $n$-dimensional Hilbert space $\h$. The matrix representation $M$ of $\E$ is an $n^2\times n^2$ matrix $M=\sum_i E_i\otimes E_i^*$ \cite{YYFD11,Wolf12}. Let $M=SJS^{-1}$ be the Jordan decomposition of $M$ where $$J=\bigoplus_{k=1}^K J_k(\lambda_k),$$ and $J_k(\lambda_k)$ is a Jordan block corresponding to the eigenvalue $\lambda_k$. Define $$J_\infty=\bigoplus_{ k:\lambda_k =1
} J_k(\lambda_k)$$ and $M_\infty = S J_\infty S^{-1}$. Then from \cite[Proposition 6.3]{Wolf12}, we know that $M_\infty$ is exactly the matrix representation of $\E_\infty$.

\begin{lemma}\label{lem:algo}
Let $\<\hs, \E\>$ be  a quantum Markov chain with $\dim(\h)=n$, and $\rho\in \D(\h)$.
\begin{enumerate}
\item The asymptotic average of $\rho$ under $\E$, i.e. $\E_{\infty}(\rho)$, can be computed in time $O(n^8)$.
\item A density operator basis of the set $\{A\in \B(\hs) : \E(A) = A\}$ can be computed in time $O(n^6)$.
\end{enumerate}
\end{lemma}
\begin{proof}

\begin{enumerate}
\item  Note that the time complexity of Jordan decomposition is
$O(d^4)$ for a $d\times d$ matrix. We can compute $M_\infty$, the matrix representation of $\E_\infty$, in time $O(n^8)$.
Then $\E_{\infty}(\rho)$ can be easily derived from the correspondence
$$(\E_{\infty}(\rho)\otimes I_\h) |\Psi\> = M_\infty(\rho\otimes I_\h) |\Psi\>$$
where $|\Psi\> =\sum_{i=1}^n |i\>|i\>$ is the (unnormalised) maximally entangled state in $\h\otimes \h$.
\item We compute the density operator basis by the following three steps:
\begin{enumerate}
    \item[(i)] Compute the matrix representation $M$ of $\E$. The time complexity is $O(mn^4)$, where $m\leq n^2$ is the number of $E_i$'s in $\E=\sum_i E_i\cdot E_i^\dag$.
    \item[(ii)] Find a basis $\B$ for the null space of the matrix $M-I_{\hs\otimes \hs}$, and transform them into matrix forms. This can be done by Guassian elimination with complexity being $O((n^2)^3) = O(n^6)$.
    \item[(iii)] For each basis matrix $A$ in $\B$, compute positive matrices $X_+,X_-,Y_+,Y_-$ such that  $\supp(X_+)\perp\supp(X_-)$, $\supp(Y_+)\perp\supp(Y_-)$, and $A = X_+-X_-+i(Y_+-Y_-)$. Let $Q$ be the set of nonzero elements in $\{X_+,X_-,Y_+,Y_-\}$. Then
by \cite[Proposition 6.8]{Wolf12}, every element of $Q$ is a fixed point state of $\E$. Replace $A$ by elements of $Q$ after normalisation. Then the resultant $\B$ is the required density operator basis. At last, we make the elements in $\B$ linearly independent. This can be done by removing elements in $\B$ using Guassian elimination.
The computational complexity of this step is $O(n^6)$.
\end{enumerate}
With these steps, we know the total complexity is $O(n^6)$. \qed
\end{enumerate}
\end{proof}

\renewcommand{\proofname}{Proof of Lemma \ref{lem:TransitiveRelation}}
\begin{proof} We only prove part 1; for the proof of part 2, see~\cite{YY12}.
    It follows from $$\supp(\rho)\subseteq \mathcal{R}_\mathcal{G}(\sigma)=\bigvee_{n=0}^\infty\supp(\E^n(\sigma))$$ that
    \begin{align*}
        \mathcal{R}_\mathcal{G}(\rho)&=\bigvee_{m=0}^\infty\E^m(\supp(\rho))\subseteq \bigvee_{m=0}^\infty\E^m(\bigvee_{n=0}^\infty\supp(\E^n(\sigma)))\\
        &=\bigvee_{i=0}^\infty\supp(\E^i(\sigma))
        =\mathcal{R}_\mathcal{G}(\sigma).
    \end{align*}
    \qed
\end{proof}

\renewcommand{\proofname}{Proof of Theorem \ref{thm:SPnondec}}

\begin{proof}
    Let $X^\perp$ be the ortho-complement of $X$ and $Q$ the projection onto $X^\perp$. Then $P_X+Q=I$.
    Let $\E^{*} = \sum_i E_i^\dag\cdot E_i$ be the Schr\"odinger-Heisenberg dual of $\E$.
    Since $X$ is invariant under $\E$, we have $\<\psi|E_i|\phi\> = 0$ for any $|\phi\> \in X$, $|\psi\> \in X^\perp$ and for any $i$. Thus $\E^{*}(Q)$ is in $X^\perp$. Furthermore, it holds that $\E^{*}(Q)\leq \E^{*}(I) = I$. This implies $\E^{*}(Q)\leq Q$. Finally, we have
    \begin{equation*}\begin{split}\tr(P_X\E(\rho)) &= \tr(\E^{*}(P_X)\rho) \\ &= \tr(\E^*(I)\rho) - \tr(\E^*(Q)\rho) \geq \tr(\rho)-\tr(Q\rho) = \tr(P_X\rho).
    \end{split}\end{equation*}
     \qed
\end{proof}

\renewcommand{\proofname}{Proof of Lemma \ref{lem:BSCCRG1}}
\begin{proof}
    We only prove the necessity part; the sufficiency part is obvious. Suppose $X$ is a BSCC. By the strong connectivity of $X$, we have $\RG(\psi)\supseteq X$ for all $|\psi\> \in X$. On the other hand, we have from the invariance of $X$ that $\E(X) \subseteq X$. Thus $\RG(\phi) = X$ for any non-zero vector $|\phi\>$ in $X$.\qed
\end{proof}

\renewcommand{\proofname}{Proof of Lemma \ref{lem:BSCCRG}}
\begin{proof} Part 1: Suppose conversely that there exists a nonzero vector $|\phi\> \in A \cap B$. Then by Lemma \ref{lem:BSCCRG1}, we have $A = \RG(\phi) = B$, contradicting the assumption that $A\neq B$. Therefore $A \cap B = \{0\}$.

Part 2: We postpone this part after the proof of Theorem~\ref{thm:Hdec}.
\qed
\end{proof}

Before proving Theorem~\ref{thm:BSCC}, we recall some basic properties of fixed point states. For more details, we refer to \cite{BCGPY12,Wolf12}.

\begin{lemma}(\cite[Proposition 6.3, Proposition 6.9]{Wolf12})\label{lem:EinftyCPTP}
    If $\E$ is a super-operator on $\hs$, then
\begin{enumerate}
\item for any density operator $\rho$, $\E_\infty(\rho)$ is a fixed point state of $\E$;
\item for any fixed point state $\sigma$, it holds that $\supp(\sigma) \subseteq \E_\infty(\hs).$
\end{enumerate}
\end{lemma}

An operator (not necessarily a partial density operator) $A\in \mathcal{B}(\h)$ is called a fixed point of super-operator $\E$ if $\E(A) = A$.

\begin{lemma}\label{lem:fpProperty}
    Let $\E$ be a super-operator on $\hs$. Then
    \begin{enumerate}
\item (\cite[Proposition 6.8]{Wolf12}) If $A$ is a fixed point of $\E$, and $$A = (X_+-X_-) + i(Y_+-Y_-)$$ where $X_+$, $X_-$, $Y_+$, $Y_-$ are positive operators with $\supp(X_+)\bot \supp(X_-)$ and $\supp(Y_+)\bot \supp(Y_-)$, then $X_+$, $X_-$, $Y_+$, $Y_-$ are all fixed points of $\E$.
        \item (\cite[Lemma 2]{BCGPY12}) If $\rho$ is a fixed point state for $\E$, then $supp(\rho)$ is an invariant subspace. Conversely, if $X$ is an invariant subspace of $\E$, then there exists a fixed point state $\rho_X$ such that $supp(\rho_X)\subseteq X$.
    \end{enumerate}
\end{lemma}
%\end{proof}

\renewcommand{\proofname}{Proof of Theorem~\ref{thm:BSCC}}

\begin{proof}
We first prove the sufficiency part. Let $\rho$ be a minimal fixed point state such that $supp(\rho) = X$. Then by Lemma \ref{lem:fpProperty}.2, $X$ is an invariant subspace. To show that $X$ is a BSCC, by Lemma~\ref{lem:BSCCRG1} it suffices to prove for any $|\phi\> \in X$, $\RG(\phi)=X$. Suppose conversely there exists $|\psi\> \in X$ such that $\RG(\psi)\subsetneq X$. Then by Lemma~\ref{lem:TransitiveRelation} $\RG(\psi)$ is an invariant subspace of $\E$ as well. By Lemma \ref{lem:fpProperty}.2, we can find a fixed point state $\rho_\psi$ with $\supp(\rho_\psi)\subseteq \RG(\psi)\subsetneq X$, contradicting the assumption that $\rho$ is minimal.

For the necessity part, let $X$ be a BSCC. Then $X$ is invariant, and by Lemma \ref{lem:fpProperty}.2, we can find a minimal fixed point state $\rho_X$ with $\supp(\rho_X)\subseteq X$. Take $|\phi\>\in \supp(\rho_X)$. By Lemma \ref{lem:BSCCRG} we have $\RG(\phi) = X$. But on the other hand, by Lemma \ref{lem:fpProperty}.2 again we have $\supp(\rho_X)$ is invariant, so $\RG(\phi) \subseteq \supp(\rho_X)$. Thus $\supp(\rho_X)= X$ indeed.

Finally, the uniqueness of $\rho$ comes from the observation that whenever $\rho$ and $\sigma$ are both fixed point states of $\E$, then so is $\lambda \rho + \gamma\sigma$, provided that it is nonzero, for any real numbers $\lambda$ and $\gamma$.
    \qed
\end{proof}

\renewcommand{\proofname}{Proof of Theorem \ref{thm:algcheckBSCC}}
\begin{proof}
By Theorem \ref{thm:BSCC}, to check whether a subspace $X$ is a BSCC, it is sufficient to check the following two properties:
\begin{enumerate}
\item $X$ is invariant;
\item Every fixed point state of $\E|_X$ has $X$ as its support. \end{enumerate}
That justifies the correctness of Algorithm \ref{alg:checkBSCC}. The complexity of this algorithm mainly comes from the statement $(*)$ which, according to Lemma~\ref{lem:algo}.2, can be computed in time $O(n^6)$. \qed
\end{proof}

\renewcommand{\proofname}{Proof of Theorem \ref{thm:Flimit}}
\begin{proof} Let $P$ be the projection onto $T_\E$, $\rho\in \D(\hs)$, and $p_k =\tr( P\E^k(\rho))$. Since $\E_\infty (\hs)$ is an invariant subspace, by Theorem \ref{thm:SPnondec} we have $p_k$ is non-increasing. Thus the limit $p_\infty = \lim_{k\rightarrow \infty} p_k$ does exist. Furthermore, noting that $\supp(\E_\infty(\rho))\subseteq \E_\infty (\hs)$, we have
    \begin{align*}
        0 &= \tr(P\E_\infty(\rho)) = \tr\left(P \lim_{N\rightarrow \infty} \frac{1}{N} \sum_{n=1}^N \E^n(\rho)\right)\\
        &= \lim_{N\rightarrow \infty}  \frac{1}{N} \sum_{n=1}^N \tr(P\E^n(\rho)) = \lim_{N\rightarrow \infty}  \frac{1}{N} \sum_{n=1}^N p_n\\
        &\geq \lim_{N\rightarrow \infty}  \frac{1}{N} \sum_{n=1}^N  p_\infty = p_\infty.
    \end{align*}
    Thus $p_\infty = 0$, and $T_\E$ is transient by the arbitrariness of $\rho$.

To show that $T_\E$ is the largest transient subspace of $\G$, note that $\supp(\E_\infty(I))=\E_\infty(\hs)$. Let $\sigma = \E_\infty(I/d)$. Then by Lemma \ref{lem:EinftyCPTP}, $\sigma$ is a fixed point state with $\supp(\sigma)=T_\E^\perp$. Suppose $Y$ is a transient subspace. We have $$\lim_{i\ra \infty}\tr(P_Y\E^i(\sigma))=\tr(P_Y\sigma) =0.$$ This implies $Y\perp\supp(\sigma)=T_\E^\perp$. That is $Y\subseteq T_\E$.
    \qed
\end{proof}

\renewcommand{\proofname}{Proof}
\begin{lemma}\label{lem:FPSdirsum}
Let $\rho$ and $\sigma$ be two fixed point states of $\E$, and $\supp(\sigma)\subsetneqq \supp(\rho)$. Then there exists another fixed point state $\eta$ with $\supp(\eta)\bot\supp(\sigma)$ such that $$\supp(\rho) = \supp(\eta)\oplus \supp(\sigma).$$
\end{lemma}
\begin{proof}
Note that for any $\lambda>0$, $\rho-\lambda \sigma$ is again a fixed point of $\E$. We can take $\lambda$ sufficiently large such that
$\rho - \lambda\sigma = \Delta_+ - \Delta_-$ with $\Delta_{\pm}\geq 0$, $\supp(\Delta_-) = \supp(\sigma)$, and
 $\supp(\Delta_+)$ is the orthogonal complement of $\supp(\Delta_-)$ in $\supp(\rho)$. By Lemma~\ref{lem:fpProperty}.1, both $ \Delta_+$ and $\Delta_-$ are again fixed point states of $\E$. Let $\eta = \Delta_+$. We have $$\supp(\rho) = \supp(\rho - \lambda\sigma) = \supp(\Delta_+)\oplus \supp(\Delta_-)=\supp(\eta)\oplus \supp(\sigma).$$
    \qed
\end{proof}
%Note: this lemma is different to Corollary 1 of \cite{BCGPY12} (Lemma \ref{lem:fpProperty}.4), because it requires the direct sum of the two orthogonal subspaces forms the original one. And it is also different to Lemma 8 of \cite{Rosmanis12}, because it does not require $\E^2=\E$.

%
%\begin{theorem}\label{thm:FPdec}
%    Let a density operator $\rho$ be a fixed point, then there exists a set of pairwise orthogonal BSCCs, $\{\bscc_i: i = 1,\cdots, k\}$, such that
%    \begin{equation*}
%        supp(\rho) = \oplus_{i=1}^k \bscc_i.
%    \end{equation*}
%\end{theorem}
%\begin{proof}
%    We apply Lemma \ref{lem:FPSdirsum} again and again to the each new proper invariant subspace. Since $\hs$ is of finit dimension, this process will terminate. And when it terminates, there is only one fixed point state of full rank in each final invariant subspace. Thus all the subspaces are BSCCs and pairwise orthogonal.\qed
%\end{proof}

The following corollary is immediately from Lemma~\ref{lem:FPSdirsum}.
\begin{corollary}\label{cor:dec} Let $\rho$ be a fixed point state of $\E$. Then $\supp(\rho)$ can be decomposed into the direct sum of some orthogonal BSCCs.
\end{corollary}
\renewcommand{\proofname}{Proof}
\begin{proof}
If $\rho$ is minimal, then by Theorem~\ref{thm:BSCC}, $\supp(\rho)$ is itself a BSCC and we are done. Otherwise, we apply Lemma \ref{lem:FPSdirsum} to obtain two fixed point states of $\E$ with smaller orthogonal supports. Repeating this procedure, we will  get a set of minimal fixed point states $\rho_1, \cdots, \rho_k$ with mutually orthogonal supports, and $$\supp(\rho) = \bigoplus_{i=1}^k \supp(\rho_i).$$ Finally, from Lemma~\ref{lem:fpProperty}.2 and Theorem~\ref{thm:BSCC}, each $\supp(\rho_i)$ is a BSCC.  \qed
\end{proof}

\renewcommand{\proofname}{Proof of Theorem \ref{thm:Hdec}}
\begin{proof}
Direct from Corollary~\ref{cor:dec} by noting that $\E_\infty(I)$ is a fixed point state of $\E$ with $\supp(\E_\infty(I)) = \E_\infty(\hs)$.     \qed
\end{proof}

\renewcommand{\proofname}{Proof of Lemma \ref{lem:BSCCRG}.2}
\begin{proof} Suppose without loss of generality that $\dim (\bscc_1)< \dim(\bscc_2)$. Let $\rho$ and $\sigma$ be the minimal fixed point states corresponding to $\bscc_1$ and $\bscc_2$, respectively. Then from Theorem~\ref{thm:BSCC}, $\supp(\rho)=\bscc_1$ and $\supp(\sigma)=\bscc_2$.
Similar to the proof of Lemma~\ref{lem:FPSdirsum}, we can take $\lambda$ sufficiently large such that
$\rho - \lambda\sigma = \Delta_+ - \Delta_-$ with $\Delta_{\pm}\geq 0$, $\supp(\Delta_-) = \supp(\sigma)$, and
 $\supp(\Delta_+)\bot \supp(\Delta_-)$. Let $P$ be the projection onto $\bscc_2$.
Then $$P\rho P = \lambda P\sigma P +  P\Delta_+P - P\Delta_- P= \lambda \sigma - \Delta_- $$ is a fixed point state as well.
Note $\supp(P\rho P) \subseteq \bscc_2$ and the fact that $\sigma$ is the minimal fixed point state corresponding to $\bscc_2$. It follows that
$P\rho P = p\sigma$ for some $p\geq 0$. Now if $p>0$, then by Proposition \ref{prop:sup}.3 we have
    \begin{equation*}
        \bscc_2 = \supp(\sigma) = \supp(P\rho P) = \spa\{P |\psi\>:|\psi\>\in \bscc_1\}.
    \end{equation*}
    This implies $\dim (\bscc_2)\leq \dim(\bscc_1)$, contradicting our assumption. Thus we have $P\rho P = 0$, which implies $\bscc_1\bot\bscc_2$.
    \qed
\end{proof}

\renewcommand{\proofname}{Proof of Theorem \ref{thm:BSCCUnitary}}
\begin{proof}
Let $b_i = \dim(\bscc_i)$, and $d_i = \dim(D_i)$. We prove by induction that $b_i = d_i$ for any $1\leq i\leq \min\{u,v\}$. Thus $u=v$ as well.

First, we claim $b_1 = d_1$. Otherwise let, say, $b_1<d_1$. Then $b_1< d_j$ for any $j$. Thus by Lemma \ref{lem:BSCCRG}.2, we have $$\bscc_1\bot \bigoplus_{j=1}^v D_j.$$ But we also have $\bscc_1 \bot T_\E$, a contradiction as
$$\bigoplus_{j=1}^v D_j \bigoplus T_\E = \hs.$$

Suppose we have $b_i = d_i$ for any $i<n$. We claim $b_n = d_n$ as well. Otherwise let, say, $b_n<d_n$. Then from Lemma \ref{lem:BSCCRG}.2, we have $$\bigoplus_{i=1}^n \bscc_i\bot \bigoplus_{i=n}^v D_i,$$ and hence $$\bigoplus_{i=1}^n \bscc_i \subseteq
\bigoplus_{i=1}^{n-1} D_i.$$ On the other hand, we have
$$\dim(\bigoplus_{i=1}^n \bscc_i ) = \sum_{i=1}^n b_i > \sum_{i=1}^{n-1} d_i = \dim(\bigoplus_{i=1}^{n-1} D_i),$$
a contradiction.
%    Let $\{|\psi_{i,j}\>\}$ be the orthonormal basis in $\bscc_i$ and $\{|\phi_{i,j}\>\}$ be the orthogonal basis in $D_i$. Then we have $\{|\psi_{i,j}\>: i = 1,\cdots,s ~~\text{and}~~ j=1,\cdots,d_i \}$  and $\{|\phi_{i,j}\>: i = 1,\cdots,s ~~\text{and}~~ j=1,\cdots,d_i \}$ are two orthonormal basis in $F$. Thus we can choose an unitary operator mapping one to the other. This complete the proof.
    \qed
\end{proof}

\renewcommand{\proofname}{Proof of Theorem~\ref{thm:algdecomposition}}
\begin{proof}
%First we give Procedure \ref{alg:decomposeA}. This procedure first decompose a fixed point $A$ into the form $A = X_+-X_-+i(Y_+-Y_-)$, such that $X_+,X_-,Y_+,Y_-\geq 0$ and $\supp(X_+)\perp\supp(X_-)$, $\supp(Y_+)\perp\supp(Y_-)$. Then it returns the set $R = \{\rho: \rho = Z/\tr(Z),~\forall Z \in Q, \text{and} ~\tr(Z)\neq 0\}$, where $Q=\{X_+,X_-,Y_+,Y_-\}$. The correctness of this procedure is guaranteed by Lemma \ref{lem:fpProperty}.3. And its computational complexity is $O(n^3)$, where $A$ is an $n\times n$ matrix.

The correctness of Algorithm~\ref{al2} follows from Theorem~\ref{thm:Flimit}, Lemma~\ref{lem:fpProperty} and Corollary~\ref{cor:dec}. For the time complexity, note that similar to Algorithm \ref{alg:checkBSCC}, the non-recursive part of the procedure $Decompose(X)$ runs in time $O(n^6)$. Thus its total complexity is $O(n^7)$, as the procedure calls itself at most $O(n)$ times.

Algorithm~\ref{al2} first computes $\E_\infty(\h)$, which costs time $O(n^8)$, by Lemma~\ref{lem:algo}.1, and then feeds it into the procedure $Decompose(X)$. Thus the total complexity of Algorithm~\ref{al2} is $O(n^8)$. \qed
 \end{proof}

\renewcommand{\proofname}{Proof of Lemma~\ref{lem:BSCClimit}}
\begin{proof}
Let $P$ be the projection onto $T_\E = \E_\infty(\hs)^\perp$.
Similar to the proof of Theorem \ref{thm:Flimit}, we have $$q_i = \lim_{k\rightarrow \infty} \tr( P_{B_i}\E^k(\rho) )$$ does exist, and $\tr(P_{B_i}\E_\infty (\rho))\leq q_i$. Moreover
    \begin{equation*}
        1 =\tr((I-P)\E_\infty(\rho)) = \sum_i  \tr( P_{B_i}\E_\infty(\rho) ) \leq \sum_i q_i = \lim_{k\ra \infty}\tr((I-P)\E^k(\rho))=1.
    \end{equation*}
    This implies $q_i = \tr( P_{B_i}\E_\infty(\rho) )$.\qed
\end{proof}

\renewcommand{\proofname}{Proof of Theorem~\ref{thm:reach}}
\begin{proof}  The claim that $$\Pr(\rho\vDash \Diamond G)  =  \tr(P_{G}\widetilde{\E}_\infty (\rho))$$ is directly from Lemma~\ref{lem:BSCClimit} and Theorem~\ref{thm:Hdec}, and the time complexity of computing this quantity follows from Lemma~\ref{lem:algo}.1.
    \qed
\end{proof}

\renewcommand{\proofname}{Proof of Lemma \ref{lem:weakFairness}}
\begin{proof}
As $X$ is not orthogonal to $\bscc$, we can always find a pure state $|\phi\>\in \bscc$ such that $P_X|\phi\>\neq 0$. Now for any $|\psi\>\in \bscc$, if there exists $N$ such that $\tr(P_X\E^{k}(\psi))=0$ for any $k>N$. Then
 $|\phi\>\not\in \RG(\E^{N+1}(\psi))$, which means that $\RG(\E^{N+1}(\psi))$ is a proper invariant subspace of $\bscc$. This contradicts the assumption that $\bscc$ is a BSCC. Thus we have $\tr(P_X\E^{i}(\psi))>0$ for infinitely many $i$.
    \qed
\end{proof}

\renewcommand{\proofname}{Proof of Lemma \ref{lem:strongFairness}}
\begin{proof}
Similar to Proposition 6.2 in \cite{Wolf12} and Lemma 4.1 in \cite{YYFD11}, we can show that the limit $$\widetilde{\E}_\infty := \lim_{N\rightarrow \infty}\frac{1}{N}\sum_{n=1}^N\widetilde{\E}^n$$ exists as well. For any $|\psi\>\in \bscc$, we claim that $\rho_\psi:=\widetilde{\E}_\infty(\psi)$ is a zero operator. Otherwise, it is easy to check that $\rho_\psi$ is a fixed point of $\widetilde{\E}$. Furthermore, from the fact that
$$\E(\rho_\psi) = \widetilde{\E}(\rho_\psi) + P_X\E(\rho_\psi)P_X=\rho_\psi + P_X\E(\rho_\psi)P_X,$$ we have $\tr(P_X\E(\rho_\psi))=0$ as $\E$ is trace-preserving. Thus $P_X\E(\rho_\psi)P_X=0$, and $\rho_\psi$ is also a fixed point of $\E$. Note that $\supp(\rho_\psi)\subseteq X^\perp \cap B$. This contradicts with the assumption that $B$ is a BSCC, by Theorem~\ref{thm:BSCC}.

With the claim, and the fact that $\tr(\widetilde{\E}^i(\psi))$ is non-increasing in $i$, we immediately have $\lim_{i\rightarrow \infty} \tr(\widetilde{\E}^i(\psi)) = 0$.
    \qed
\end{proof}

\renewcommand{\proofname}{Proof of Theorem \ref{thm:Fairness}}
\begin{proof}
Similar to the proof of Lemma \ref{lem:strongFairness}.
    \qed
\end{proof}

To prove Theorem~\ref{thm:xGyG}, we need the following lemma.
\begin{lemma}\label{lem:FsubInv}
Suppose $\E_\infty{(\hs)}$ has a proper invariant subspace $S$. Then for any density operator  $\rho$ with $\supp(\rho)\subseteq \E_\infty{(\hs)}$ and any integer $k$, we have
    \begin{equation*}
        \tr(P_S\E^k(\rho)) = \tr(P_S\rho)
    \end{equation*}
    where $P_S$ is the projection onto $S$.
\end{lemma}
\renewcommand{\proofname}{Proof}
\begin{proof}
    By Lemma \ref{lem:FPSdirsum}, there exists an invariant subspace $T$ such that $\E_\infty{(\hs)}=S\oplus T$ where $S$ and $T$ are orthogonal and invariant. Then by Theorem \ref{thm:SPnondec}, we have \begin{equation*}
        1\geq\tr(P_S\E^k(\rho))+\tr(P_T\E^k(\rho))\geq\tr(P_S\rho)+\tr(P_T\rho)=\tr(\rho)=1.
    \end{equation*}
    Thus $\tr(P_S\E^k(\rho)) = \tr(P_S\rho)$.
    \qed
\end{proof}

\renewcommand{\proofname}{Proof of Theorem \ref{thm:xGyG}}
 \begin{proof}
We first show that $\mathcal{Y}(G)$ is a subspace. Let $|\psi_i\>\in \mathcal{Y}(G)$ and $\alpha_i$ be complex numbers, $i=1,2$. Then there exists $N_i$ such that for any $j\geq N_i$,  $\supp(\E^j(\psi_i))\subseteq G$. Let $|\psi\> = \alpha_1|\psi_1\> + \alpha_2 |\psi_2\>$ and $\rho =|\psi_1\>\<\psi_1| + |\psi_2\>\<\psi_2|$. Then $|\psi\> \in \supp(\rho)$, and from Propositions~\ref{prop:sup}.1, \ref{prop:sup}.2, and \ref{prop:sup}.4 we have $$\supp(\E^j(\psi))\subseteq \supp(\E^j(\rho))=\supp(\E^j(\psi_1))\vee\supp(\E^j(\psi_2))$$ for any $j\geq 0$. Now $|\psi\>\in \mathcal{Y}(G)$ follows by letting $N=\max\{N_1, N_2\}$.

 We divide the rest of proof into six parts.
\begin{itemize}
\item Claim 1: $\mathcal{Y}(G)\supseteq \bigvee_{\bscc\subseteq G} \bscc.$

\smallskip\

For any BSCC $\bscc\subseteq G$, from Lemmas~\ref{lem:EinftyCPTP}.2 and \ref{lem:fpProperty}.2 we have $\bscc\subseteq \E_{\infty}(\h)$. Furthermore, as $\bscc$ is a BSCC, for any $|\psi\>\in \bscc$ and any $i$, $\supp(\E^i(\psi))\subseteq \bscc\subseteq G$. Thus $B\subseteq \mathcal{Y}(G)$, and the result follows from the fact that $\mathcal{Y}(G)$ is a subspace.

\smallskip\

\item Claim 2: $\mathcal{Y}(G)\subseteq \bigvee_{\bscc\subseteq G} \bscc$.

\smallskip\

For any $|\psi\>\in \mathcal{Y}(G)$, note that $\rho_\psi := \E_\infty(\psi)$ is a fixed point state. Let $X=\supp(\rho_\psi)$. We claim $|\psi\>\in X$. This is obvious if $X=\E_\infty(\h)$. Otherwise, as $\E_\infty(I_\h)$ is a fixed point state and
$\E_\infty(\h)=\supp(\E_\infty(I_\h))$, by Lemma \ref{lem:FPSdirsum} we have
$\E_\infty(\h)=X\oplus X^\perp$, where $X^\perp$, the ortho-complement of $X$ in $\E_\infty(\h)$, is also invariant. As $X$ is again a direct sum of some orthogonal BSCCs, by Lemma~\ref{lem:BSCClimit} we have $$\lim_{i\ra \infty}\tr(P_X\E^i(\psi)) = \tr(P_X\E_\infty(\psi))= 1;$$ that is, $$\lim_{i\ra\infty}\tr(P_{X^\perp}\E^i(\psi))=0.$$ Together with Theorem \ref{thm:SPnondec}, this implies $\tr(P_{X^\perp}\psi)=0$, and so $|\psi\>\in X$.

By the definition of $\mathcal{Y}(G)$, there exists $M\geq 0$, such that $\supp(\E^i(\psi))\subseteq G$ for all $i\geq M$. Thus
    \begin{align*}
        X & = \supp(\lim_{N\ra\infty}\frac{1}{N}\sum_{i=1}^N\E^i(\psi))= \supp(\lim_{N\ra\infty}\frac{1}{N}\sum_{i=M}^N\E^i(\psi))\subseteq G.
    \end{align*}
Furthermore, since $X$ can be decomposed into the direct sum of some BSCCs, we have $X\subseteq \bigvee_{\bscc\subseteq G} \bscc$.
Then the result follows by noting $|\psi\>\in X$.

\smallskip\

\item Claim 3: $\mathcal{Y}(G^\perp)^\perp\subseteq \mathcal{X}(G)$.

\smallskip\

First, from Claims 1 and 2 above we have
$\mathcal{Y}(G^\perp)\subseteq G^\perp$ and $G':=\mathcal{Y}(G^\perp)^\perp$ is invariant.
Thus $G\subseteq \mathcal{Y}(G^\perp)^\perp$, and $\E$ is a also super-operator on $G'$; that is, the pair
$\<G', \E\>$ is again a quantum Markov chain. Furthermore, Claim 1 implies that any BSCC in $G^\perp$ is also contained in $\mathcal{Y}(G^\perp)$. Thus there is no BSCC in $G'\cap G^\perp$. By Theorem \ref{thm:Fairness}, for any $|\psi\>\in G'$, $\lim_{i\ra\infty}\tr[(P_{G^\perp}\circ\E)^i(\psi)] = 0$. Thus $|\psi\>\in \mathcal{X}(G)$ by definition.

\smallskip\

\item Claim 4: $\mathcal{X}(G) \subseteq\mathcal{Y}(G^\perp)^\perp$.

\smallskip\

Similar to Claim 3, we have
$\mathcal{Y}(G^\perp)\subseteq G^\perp$ and $\mathcal{Y}(G^\perp)$ is invariant.
 Let $P$ be the projection onto $\mathcal{Y}(G^\perp)$. Then $P_{G^\perp} P P_{G^\perp}=P$. For any $|\psi\>\in  \mathcal{X}(G)$, we calculate
   \begin{equation*}
      \tr(P(P_{G^\perp}\circ \E)(\psi))=  \tr(P_{G^\perp} P P_{G^\perp} \E(\psi)) = \tr(P\E(\psi)) \geq \tr(P\psi),
    \end{equation*}
    where the last inequality is by Theorem \ref{thm:SPnondec}. Therefore
    \begin{equation*}
        0 = \lim_{i\ra\infty}\tr((P_{G^\perp}\circ \E)^i(\psi)) \geq \lim_{i\ra\infty}\tr(P(P_{G^\perp}\circ \E)^i(\psi)) \geq \tr(P\psi),
    \end{equation*}
    and so $|\psi\>\in \mathcal{Y}(G^\perp)^\perp$.

    \smallskip\

  \item  Claim 5: $\bigvee_{\bscc\subseteq G} \bscc\subseteq \E_\infty(G^\perp)^\perp$.

  \smallskip\

Suppose a BSCC $\bscc\subseteq G$. Then we have $\tr(P_BI_{G^\perp}) = 0$, and so $\tr(P_B\E^i(I_{G^\perp}))=0$ for any $i\geq 0$ by Lemma \ref{lem:FsubInv}. Thus $\tr(P_B\E_\infty(I_{G^\perp}))=0$, leading to $\bscc\perp\E_\infty(G^\perp)$. Therefore $\bscc\subseteq \E_\infty(G^\perp)^\perp$. Then the result follows from the fact that $\E_\infty(G^\perp)^\perp$ is a subspace.

\smallskip\

\item  Claim 6: $\E_\infty(G^\perp)^\perp \subseteq \bigvee_{\bscc\subseteq G} \bscc$.

\smallskip\

    By Corollary~\ref{cor:dec}, $\E_\infty(G^\perp)^\perp$ can be decomposed into direct sum of BSCCs $\bscc_1,\bscc_2,\cdots$. For any $\bscc_i$, we have $\tr(P_{\bscc_i}\E_\infty(I_{G^\perp})) = 0$. Thus $\tr(P_{\bscc_i}I_{G^\perp}) = 0$, meaning that $\bscc_i\perp G^\perp$. Therefore $\bscc_i\subseteq G$, and the result holds.
\end{itemize}
The invariance of $\mathcal{X}(G)$ and $\mathcal{Y}(G)$ is already included in Claims 1 and 2. This complete the proof.\qed
\end{proof}

%\renewcommand{\proofname}{Proof of Lemma \ref{lem:conBSCC}}
%\begin{proof}
%    By Proposition \ref{prop:sup}.1, if $\E=\sum E_i\cdot E_i^\dag$ then $S$ is invariant under $\E$ if and only it is invariant under $E_i$ for any $i$. Thus the lemma immediately follows.  \qed
%\end{proof}
%
%\renewcommand{\proofname}{Proof of Theorem \ref{thm:conPers}}
%\begin{proof}
%    $\E$ is a CPTP map, then $F$ can be decomposed into $F=\oplus_{i=1}^u \bscc_i$, where $P(G) =\oplus_{i=1}^v \bscc_i$, $(v<u)$ and $\bscc_i$s are pairwise orthogonal BSCCs under $\E$ and invariant under $\E_i$ for any $i$. Let $P$ be a projection onto $P(G)$ and $Q$ be a projection onto $\oplus_{i=v+1}^u \bscc_i$.
%
%    Suppose $\rho_0\subseteq P(G)$. By Lemma \ref{lem:FsubInv}, we have $\tr(Q\rho_k)=0$, for any execution $\pi$ and $\rho_{k} = \E_{s_{k}}(\rho_{k-1})$. Thus $\supp(\rho_k)\subseteq P(G)\subseteq G$, meaning that $(\rho_0,\pi)\vDash \text{pers}(G)$ for any $\pi$.
%    Conversely if $(\rho_0,\pi)\vDash \text{pers}(G)$ for any $\pi$, we claim $\supp(\rho_0)\subseteq P(G)$. Otherwise $\tr(Q\rho_0)>0$. Then $\tr(Q\rho_k)>0$ for any execution $\pi$. This completes the proof.
%    \qed
%\end{proof}

\renewcommand{\proofname}{Proof of Theorem~\ref{thm:algpersistence}}
\begin{proof}
%Although BSCCs are necessary in proving Theorem \ref{thm:probPR}, they will disappear when calculating the probabilities.
%Thus we put the algorithms for decomposing $\hs$ into BSCCs in the Appendix \ref{sec:alg}.
The correctness of Algorithm~\ref{alg:pers} follows from Theorems~\ref{thm:xGyG} and \ref{thm:probPR}. The time complexity is again dominated by Jordan decomposition used in computing $\E_\infty(\rho)$ and $\E_\infty(G^\perp)$, thus it is $O(n^8)$.
\qed
\end{proof}
\end{document}